\documentclass[10pt, twocolumn]{article}
\usepackage{graphicx}
\usepackage{amsmath}
\usepackage{amsfonts}
\usepackage{amssymb}
\usepackage{latexsym}
\usepackage{amsthm}

\usepackage{times}

\newcommand{\YA}{Yahoo Answers }
\newcommand{\YAc}{Yahoo Answers}
\newcommand{\comment}[1]{}

\newtheorem{hypothesis}{Hypothesis}
\newtheorem{lemma}{Lemma}

\usepackage{setspace}
\renewcommand{\P}{\ensuremath{\mathbb{P}}}
\newcommand{\E}{\ensuremath{\mathbb{E}}}

\begin{document}

\title{ Human Speed-Accuracy Tradeoffs in Search}
\author{Christina Aperjis  \\ HP Labs \\ christina.aperjis@hp.com \and Bernardo A. Huberman \\ HP Labs \\ bernardo.huberman@hp.com \and Fang Wu \\ HP Labs \\ suncube@gmail.com}
\date{}

%
%
%

\maketitle

\begin{abstract}

\emph{When foraging for information, users face a tradeoff between the accuracy and value of the acquired information and the time spent collecting it, a problem which also surfaces when seeking answers to a question posed to a large community. We empirically study how people behave when facing these conflicting objectives using data from \YAc, a community driven question-and-answer site.
We first study how users behave when trying to maximize the amount of acquired information while minimizing the waiting time. We find that users are willing to wait longer for an additional answer if they have received a small number of answers.
We then assume that users make a sequence of decisions, deciding to wait for an additional answer as long as the quality of the current answer exceeds some threshold. The resulting probability distribution for the number of answers that a question gets is an inverse Gaussian, a fact that is validated by our data.
}

\end{abstract}

\section{Introduction}

When searching for an answer to a question, people face a well known tradeoff between the accuracy of the acquired information and the time spent collecting it.  The fact that it usually takes longer to find a better answer to a given question creates a dilemma which is inherent in information seeking.  Stopping the search for information early provides swift answers which might not be completely correct, whereas continuing to search for longer usually provides accuracy and completeness while sacrificing the timeliness of the answer.  Worse, if answers to a question arrive at random intervals, a user seeking specific information that will inform a decision to be made faces the problem of having to wait for an uncertain interval of time in exchange for an answer which may or may not improve on what is already known.

Rather than study these issues in a laboratory setting, we decided to do so in the natural context of the Web. Specifically, we studied user behavior within \YAc, a community-driven question-and-answer site with more than 21 million unique users.
At \YAc, users post questions seeking to harvest the collective intelligence of others in the system.  Once a user submits a question, it gets posted on the site.  Other users can then submit answers to the question, which are also posted on the site.  When the author of a question is satisfied with the answers he has received, he closes the question and thus terminates his search for answers.  He then uses information from the answers he received to build his own ``aggregate answer'' to his question, that will potentially help him take a related decision.

There are two aspects that users value with respect to the aggregate answers they obtain: {\em accuracy} and {\em speed}, and thus they try to maximize the accuracy of their aggregate answers without waiting too long. The accuracy of the aggregate answer depends on the accuracy of all individual answers that the question received.

Anyone posting a question faces the following tradeoff at any given point in time.  He can either build his aggregate answer at a given time or wait for additional answers to arrive.  If he waits, he may achieve a higher accuracy in the future, but also incurs a cost for waiting.  The user's intent is to build his aggregate answer at the optimal stopping time.  We take two complementary approaches to the analysis of user behavior with respect to speed-accuracy tradeoffs.

Our first approach studies the speed-accuracy tradeoff by using the number of answers as a proxy for accuracy.  In particular, we assume that the user estimates the accuracy of his aggregate answer by the number of answers that his question gets.  Thus, he faces the following tradeoff: he prefers more to less answers, but does not want to wait too long.  We analyze \YA data to identify and quantify this tradeoff.  Our first finding is that users are willing to wait more to obtain one additional answer when they have only received a small number of answers; this implies  decreasing marginal returns in the number of answers.  Formally, this implies a concave utility function in the amount of information.  We then estimate the utility function from the data.

Our second approach considers the qualities of the individual answers without explicitly computing the cost of waiting.  We assume that users 
decide to wait as long as the value of the current answer exceeds some threshold.  Under this model, the probability distribution for the number of answers that a question gets is an inverse Gaussian, which is a Zipf-like distribution.  We use the data to validate this conclusion.

The rest of the paper is organized as follows.  Section \ref{sec:related} reviews related work.  Section \ref{sec:rules} describes \YA focussing on the rules that are important for our analysis.  Section \ref{sec:tradeoff} empirically studies the speed-accuracy tradeoff by using the number of answers as a proxy for accuracy.  Section \ref{sec:quality} we focuses on how users assess quality.   Section \ref{sec:conclusion} concludes.

\section{Related Work}
\label{sec:related}

This paper uses \YA data to study information seeking behavior with respect to stopping when people face speed-accuracy tradeoffs.  
In this section we review related work.

A number of papers have studied information seeking behavior on the Web.  For instance, Efthimiadis considers the search for health information~\cite{HICSS:health}, and Egusa et al. study how people search when asked to complete certain tasks~\cite{HICSS:depth}.  Russell et al. present a search-task taxonomy~\cite{HICSS:intent}.
However, these papers do not explicitly consider behavior with respect to stopping and speed-accuracy tradeoffs.

The behavior of people with respect to stopping problems has been studied extensively in the context of the secretary problem~\cite{classical}.
In the classical secretary problem, applicants are interviewed sequentially in a random order and the goal is to maximize the probability of choosing the best applicant.  The applicants can be ranked from best to worst with no ties.  After each interview, the applicant is either accepted or rejected.
If the decision maker knows the total number of applicants $n$, for large $n$ the optimal policy is to interview and reject the first $n/e$ applicants (where $e$ is the base of the natural logarithm) and then to accept the next who is better than these interviewed candidates~\cite{classical}.

Experimental studies of the classical secretary problem and variants show that people tend to stop too early and give insufficient consideration to the yet-to-be-seen applicants (e.g., ~\cite{bearden}).  On the other hand, when there are search costs and recall (backward solicitation) of previously inspected alternatives is allowed, people tend to search longer than the optimum~\cite{recall}.
We note a key difference with the setting of information seeking: while in the secretary problem only one secretary can be hired, an information seeker can combine information from multiple sources to build a more accurate answer for his question.  Moreover, in the secretary problem the decision maker does not face a speed-accuracy tradeoff, because time does not affect his payoff.

The speed-accuracy tradeoff has been considered in various settings.  One example is a setting where a group cooperates to solve a problem~\cite{speed_accuracy}.
In psychology, on the other hand, the speed-accuracy tradeoff is used to describe the tradeoff between how fast a task can be performed and how many mistakes are made in performing the task (e.g.,~\cite{impulsivity, psycho}).

There has been a number of empirical studies that use data from \YA and other question-answering communities.
Data from \YA have been used to predict whether a particular answer
will be chosen as the best answer~\cite{adamic}, and whether a user will be satisfied with the answers to his question~\cite{datasetA}.
Content analysis has been used to study the criteria with which users select the best answers to their questions~\cite{selection_criteria}.
Shah et al. study the effect of user participation on the success of a social Q\&A site~\cite{QAparticipation}.
Aji and Agichtein analyze the factors that influence how the \YA community responds to a question~\cite{datasetB}.
Finally, various characteristics of user behavior in terms of asking and answering questions have been considered in~\cite{koutrika}.
To the best of our knowledge, there have been no studies that consider user behavior in terms of the speed-accuracy tradeoff in question-answering communities.

\section{Yahoo Answers}
\label{sec:rules}

In this section we describe the rules of \YA that are important for our analysis, and explain why users often face speed-accuracy tradeoffs.  We then briefly discuss the data we use.

\subsection{Rules}

\YA is a question-and-answer site that allows users to both submit questions to be answered and answer questions asked by other users.
When a user submits a question, the question is posted on the \YA site.  Other users can then see the question and submit answers, which are also posted on the site.  According to standard \YA terminology, the user that asks the question is called the {\em asker}, and a user answering is called an {\em answerer}.  In this paper, we study the behavior of the asker, and thus the word {\em user} is used to describe the asker.

Once the user starts receiving answers to his question, he can choose the best answer at any point in time.  After the best answer to a question is selected, the question does not receive any additional answers.  We thus say that a user closes the question when he chooses the best answer.  Closing the question is equivalent to terminating the search for answers to the question.

Questions have a 4-day open period.  If a question does not receive any answers within the 4-day open period, it expires and is deleted. However, before the question expires, the asker has the option to extend the time period a question is open by four more days.  The time can only be extended once.  However, most questions are not extended, and in our analysis we only consider questions that were closed within the 4-day period.
If the asker does not choose a best answer to his question within the 4-day open period, then the question is up for voting, that is, other \YA users can vote to determine the best answer to the question.

\subsection{Askers Face Speed-Accuracy \\ Tradeoffs}

We expect that the user is satisfied with the answers he received when he closes the question.  The user then uses information from these answers to build his own aggregate answer to his question.
Throughout the paper, we use the term {\em aggregate answer} to refer to the conclusion that the question author draws by reading the answers to his question.  The aggregate answer is {\em not} posted on the \YA site, and is often not recorded.

Askers at Yahoo Answers are often asking questions to get information that will help them make a decision.  In particular, the decision will be aided by the aggregate answer.
In many cases, the asker prefers to decide sooner than later. For instance, this is usually the case when the asker is seeking information on which product to buy, because the sooner he decides what he is buying, the sooner he will get the product, and the sooner he will derive value by using it.

As an example, we can consider the following question from Yahoo Answers: ``What is the best graphics card for gamers?'' In this case, the asker specified his price range and what he wanted to use the card for. The answerers then wrote their opinions on which card would be the best for the asker. Then, the asker chose the best answer and stated his decision on which card he was going to buy. It is reasonable to expect that in this case the asker would derive strictly higher utility by having the graphics card sooner.

This motivates the speed-accuracy tradeoff. In particular, if the question has $n$ answers, the asker can either make a decision now using these $n$ answers, or make a decision later using more information. The cost of waiting is the cost of delaying the decision.

We do not suggest that an asker closes his question because of the speed-accuracy tradeoff. We just expect that if an asker did close his question, then the closing time provides a good approximation to the time that the asker made his decision.  We use the closing time in our analysis to understand how users behave with respect to speed-accuracy tradeoffs.

We next briefly discuss why an asker may close his question. There are two motivations. First, this is a way of thanking the answerer for his time and effort. Second, an asker gets some points in the Yahoo Answers reputation system if he closes his question. Clearly, not everyone is sufficiently motivated to close his answers in this way, and as a result many askers do not close their questions. However, given that an asker did close his question, we expect that the closing time gives information about how the asker behaved with respect to the speed-accuracy tradeoff.

\subsection{Data}

We use a \YA dataset that was crawled in October 2008 by Aji and Agichtein~\cite{datasetB}.
For each question in this dataset we know the time the question was posted, the arrival time of each answer to the question, and the time that the asker closed the question by selecting the best answer.

For the purposes of this paper, we only consider questions for which the best answer was selected by the asker.  The reason is that we are interested in the time that the asker terminates his search for information by closing the question.  If the asker selects the best answer, this is the time that the best answer was selected.  On the other hand, if the asker does not select a best answer, we have no relevant information (we do not know when and whether the asker built his aggregate answer).

Furthermore, we restrict attention to questions that were open for less than 100 hours.  This is motivated by the fact that questions are initially open for 4 days (96 hours) and that most askers close their questions within this 4-day period.

We thus use a subset of the originally collected data that consists of questions that were closed by the asker in less than 100 hours.
This subset consists of 1,536 questions.

One could argue that there is no reason for a user to close his question before the 4-day open period is over.  In particular, he could use the information from the answers he has received up to now, but close the question at the end of the fourth day.  However, once a user takes a decision (e.g., buys a product), there is little or no value in getting additional information.  Thus, if the user is planning to select a best answer to his question (in order to thank the answerer), he has no reason to wait until the 4-day period is over.
Indeed, Figure \ref{fig:stop} shows that most questions close a significant amount of time before the 4-day deadline.  For instance, 29\% of the questions close within just one day after the question was posted.

\begin{figure}
\begin{center}
		\includegraphics[width=0.5\textwidth]{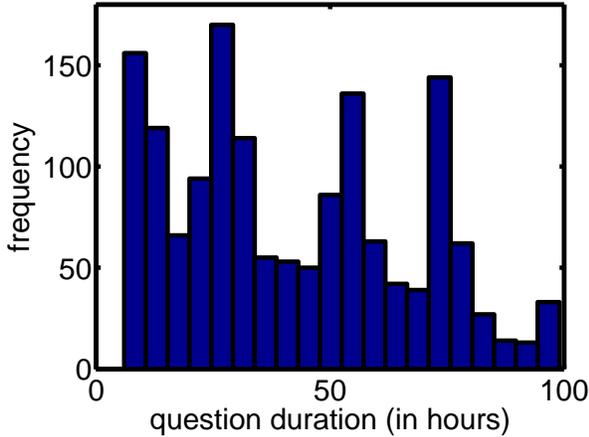}
\end{center}
\vspace{-1ex}
\caption{Histogram of the number of hours that questions were open (before the asker chose the best answer).}
\vspace{-2ex}
\label{fig:stop}
\end{figure}

\section{Speed-Quantity Tradeoff}
\label{sec:tradeoff}

A user wants to get an accurate aggregate answer (that will potentially help him take a decision) without waiting too long.  The accuracy of the user's aggregate answer is subjective and hard to measure.  In this section we use the number of answers as an approximation of accuracy.  Thus, we expect that a user's utility increases in the total number of answers that his question receives, and decreases in the time he waits for answers to arrive.

Section \ref{sec:utility} develops our hypotheses drawing on a utility model.  In Section \ref{sec:elapsed} we test our first hypothesis.  In Section \ref{sec:logit} we introduce a discrete choice model, which we estimate in Section \ref{sec:discrete} to test the remaining hypothesis.  Finally, in Section \ref{sec:ut_est} we discuss the form of the utility function.

\subsection{Utility Model}
\label{sec:utility}

Let $n$ be the total number of answers at the time the user builds his aggregate answer.  We assume that the user gets utility $u(n)$. Furthermore, we assume that the user incurs a cost $c(t)$ for waiting for time $t$.  Thus, the user is seeking to maximize $u(n) - c(t)$.

\subsubsection{Myopic Decision Rule}

Suppose that $n$ answers have arrived.  The user can either terminate his search by choosing the best answer, or wait for additional answers.  If he terminates his search now, he can build his aggregate answer using the $n$ answers that he has received, and thus get utility $u(n)$.  If he chooses to wait and a new answer arrives $t$ time units later, then he will have $n+1$ answers, but will have also incurred a cost $c(t)$ for waiting.  His utility will then be $u(n+1) - c(t)$. The user is better off stopping if $u(n+1) - u(n) < c(t)$, and continuing if $u(n+1) - u(n) > c(t)$.
In words, {\em the user decides to close the question if the cost of waiting for one more answer exceeds the incremental benefit of having one more answer}.

Our previous description assumes that the user knows when the next answer will arrive, which is not the case in reality.  
More generally, let $T$ be a random variable that describes the user's belief on how long it will take until the next answer arrives.  Then, the user is better off closing the question if $u(n+1) - u(n) < \E [c(T)]$, and continuing if $u(n+1) - u(n) > \E [c(T)]$.

The strategy we just described is myopic, since it assumes that a user decides whether to wait (i.e., not to close the question) by only considering whether he is better off waiting for one more answer.  
Alternatively, if the user knew when each answer is going to arrive in the future, we could consider a global optimization problem: if the $i$-th answer is expected to arrive at time $t_i$, the user would choose to close the question at the time $t_j$ that maximizes $u(j) - c(t_j)$.
However, in the context of \YA it is impossible for users to know when all future answers will arrive.  It is thus more realistic to assume that users myopically optimize as randomness is realized.

\begin{figure}
\begin{center}
\begin{tabular}{|l|}
  \hline
If $u(n+1) - u(n) < \E [c(T)]$, then close the question\\
\hline
If $u(n+1) - u(n) > \E [c(T)]$, then wait \\
  \hline
\end{tabular}
\end{center}
\caption{Myopic decision rule.}
\label{fig:myopic}
\end{figure}

We summarize the myopic decision rule in Figure \ref{fig:myopic}.
It implies that a user is more likely to close the question when $u(n+1) - u(n)$ is small and/or $\E [c(T)]$ is large.

\subsubsection{Hypotheses}

We next develop our hypotheses building on the myopic decision rule.
Our hypotheses can be grouped in two categories.  The first category is based on the assumption that the marginal benefit of having one more answer decreases as more answers arrive; the second considers how users estimate when the next answer will arrive.

The user's valuation for having $n$ answers is $u(n)$.  We expect that $u(n)$ is concave, i.e., the marginal benefit of having one more answer decreases as the number of answers increases.  According to the myopic decision rule (Figure \ref{fig:myopic}), the user is more likely to close his question when $u(n+1) - u(n)$ is small.  Since we expect that $u(n+1) - u(n)$ is decreasing in $n$, the user is more likely to close his question when $n$ is large, i.e., when he has already received a large number of answers.   We test this in two ways, outlined in Hypotheses \ref{h:elapsed} and \ref{h:number}.

\begin{hypothesis}
\label{h:elapsed}
The amount of time that a user waits since the arrival of the last answer before closing his question is decreasing in the number of answers that the question has received.
\end{hypothesis}

\begin{hypothesis}
\label{h:number}
A user is more likely to close his question if the question has received many answers.
\end{hypothesis}

The user believes that the time until the next answer arrives is described by some random variable (which can be degenerate if he is only using an estimate).
It is reasonable to assume that a user forms his belief using the information available to him, that is, the arrival times of previous answers and the time he has waited since the last answer arrived.

A particularly important summary statistic is the last inter-arrival time, i.e., the time between the arrivals of the two most recent answers.  The last inter-arrival time is an estimate of the inverse current arrival rate of answers.  Thus, the user may use the last inter-arrival time as an estimate of the next inter-arrival time, i.e., the time between the arrival of the last answer and the next answer.  More generally, the user may form a belief on the next inter-arrival time that depends on the last inter-arrival time in some increasing fashion.
Then, if the last inter-arrival time is large, the user expects to wait a long time until he receives another answer, thus incurring a large waiting cost.  This encourages the user to close the question now.  This is the context of Hypothesis \ref{h:interarrival}.

\begin{hypothesis}
\label{h:interarrival}
A user is more likely to close his question if the last inter-arrival time is large.
\end{hypothesis}

This hypothesis is based on the assumption that the last inter-arrival time may be used as an estimate for the next inter-arrival time.
However, if a long time has elapsed since the last answer arrived (e.g., a longer period than the last inter-arrival time), the user becomes less certain about this estimate.  The increased uncertainty may lead him to expect a greater waiting cost until the next answer arrives.  In turn, this encourages the user to close the question, as is outlined in Hypothesis \ref{h:hours}.

\begin{hypothesis}
\label{h:hours}
A user is more likely to close his question if a long time has elapsed since the most recent answer arrived.
\end{hypothesis}

Hypothesis \ref{h:elapsed} is tested in Section \ref{sec:elapsed}.  Then, in Section \ref{sec:logit} we introduce a discrete choice model, which we estimate in Section \ref{sec:discrete} to test Hypotheses \ref{h:number}, \ref{h:interarrival}, and \ref{h:hours}.

\subsection{Time Between Last Arrival and \\ Closure}
\label{sec:elapsed}

In this section, we test whether a user waits longer before closing his question when the question has received a small number of answers (Hypothesis \ref{h:elapsed}).

For every question we consider the following variables:
\begin{itemize}
\item TotalAnswers: the total number of answers that the question received.  This is the number of answers at the time that the asker closed the question.
\item ElapsedTime: the time between the arrival of the last answer and the time the user closed the question.
\end{itemize}
We test for correlation between TotalAnswers and ElapsedTime.
The results are presented in Table \ref{tab:cor}.
We find that TotalAnswers and ElapsedTime are negatively correlated, bringing support for Hypothesis \ref{h:elapsed}.

\begin{table}
\begin{center}
\begin{tabular}{|c|c|c|}
   \hline
  Corelation coefficient   &  -0.148***\\
  95\% confidence interval  &  [-0.196, -0.098]\\
  \hline
  Observations &  1,536\\
  \hline
\end{tabular}
\end{center}
\caption{Correlation between the number of answers (TotalAnswers) and the time that the user waits before closing the question (ElapsedTime).  *, ** and *** denote significance at 1\%, 0.5\% and 0.1\% respectively.}
\label{tab:cor}
\end{table}

Table \ref{tab:cor} suggests that the user is willing to wait more (and incur more cost from waiting) for an additional answer if only a few answers have arrived up to now.  This implies that the marginal benefit of having one additional answer decreases as the number of answers increases.

\subsection{Model Specification}
\label{sec:logit}

In this section, we introduce a logit model, which is estimated in Section \ref{sec:discrete} to test Hypotheses \ref{h:number}, \ref{h:interarrival}, and \ref{h:hours}.

A user posts a question.  Then, at various points in time he revisits \YA to see the answers that his question has received, and decides whether to close the question by selecting the best answer.  We are interested in the probability that the user closes the question during a given visit.
For every visit we consider the following variables:
\begin{itemize}
\item $p$: the probability that the user closes the question during the visit.
\item $n$: the number of answers that the question has received by the time of the visit.
\item $l$: the last inter-arrival time, i.e., the time between the arrivals of the two most recent answers (at the time of the visit).
\item $w$: the time since the last answer arrived, i.e., the time that the user has been waiting for an answer since the last arrival.  This is equal to the difference between the time of the visit and the arrival time of the most recent answer.
\end{itemize}

Recall the utility model introduced in Section \ref{sec:utility}.  If the user closes the question at $n$ answers, his utility is $u(n)$.  Suppose that the user believes that the next answer will arrive in time $T$, where $T$ is some random variable.  Moreover, we assume that the user uses the last inter-arrival time $l$ and the time since the last answer $w$ to form his belief; that is $T$ depends on $l$ and $w$, and we write $T(l,w)$.
Then, the user expects to obtain utility $u(n+1) - \E [c(T(l,w))]$ from waiting.  According to the myopic decision rule (Figure \ref{fig:myopic}), the user decides whether to close the question or not depending on which of the expressions $u(n)$, $u(n+1) - \E [c(T(l,w))]$ is larger.

We now perturb $u(n)$ and $u(n+1) - \E [c(T(l,w))]$ with some noise.  In particular, we assume that the user's utility is
\[u(n) + \epsilon_0\]
if he closes the question, and
\[u(n+1) - \E [c(T(l,w))] + \epsilon_1\]
if he waits for the next answer.    Thus, the probability of closing the question is
\begin{align*}
p &= \P [u(n) + \epsilon_0 > u(n+1) - \E [c(T(l,w))] + \epsilon_1] \notag \\
         &= \P [\epsilon_1 - \epsilon_0 < \E [c(T(l,w))] - (u(n+1) - u(n))] \notag\\
         &= F(\E [c(T(l,w))] - (u(n+1) - u(n))),
         \label{eq:logut}
\end{align*}
where $F$ is the cumulative distribution function of $\epsilon_1 - \epsilon_0$.

Suppose that $\epsilon_0$ and $\epsilon_1$ are independent type 1 extreme value distributed.  Then, the difference $\epsilon_1 - \epsilon_0$ is logistically distributed, and
\[p = \Lambda(\E [c(T(l,w))] - (u(n+1) - u(n))),\]
where
\[\Lambda(z) = \frac{1}{1+e^{-z}}\]
is the logistic function.

The previous argument gives rise to the logit model, a standard discrete choice model in microecomics (see e.g., \cite{econometrics}).\footnote{If we assume a different distribution for $\epsilon_0$ and $\epsilon_1$, we get a different discrete choice model.  The logit model is widely used because of its simplicity.  Another widely used model is the probit model, which assumes that $\epsilon_0$ and $\epsilon_1$ are normally distributed.  We note that the probit model gives the same qualitative results as the logit model for out dataset.}

In Section \ref{sec:discrete} we estimate the following model:

\begin{equation}
\label{eq:logit}
p = \Lambda(\alpha + \beta_1 \cdot n + \beta_2 \cdot l + \beta_3 \cdot w),
\end{equation}
so that
\[\E [c(T(l,w))] - (u(n+1) - u(n)) = \alpha + \beta_1 \cdot n + \beta_2 \cdot l + \beta_3 \cdot w.\]
This implies that the marginal benefit of having one more answer when $n$ answers have arrived is
\begin{equation}
\label{eq:utility}
u(n+1) - u(n) = \alpha_u - \beta_1 \cdot n
\end{equation}
and the expected cost of waiting for the next answer is
\begin{equation}
\label{eq:cost}
\E [c(T(l,w))] = \alpha_c + \beta_2 \cdot l + \beta_3 \cdot w
\end{equation}
such that
\[\alpha_c - \alpha_u = \alpha.\]
Equations \eqref{eq:utility} and \eqref{eq:cost} are used in Section \ref{sec:ut_est} to interpret the estimated parameters of \eqref{eq:logit}.

\subsection{Model Estimation}
\label{sec:discrete}

We now use logistic regression to estimate \eqref{eq:logit}, i.e., we estimate the probability that a user closes his question as a function of (i) the number of answers ($n$), (ii) the last inter-arrival time ($l$), and (iii) the time that the user has waited since the last answer arrived ($w$).  We find that the probability of closing the question increases with all three variables, supporting Hypotheses \ref{h:number}, \ref{h:interarrival}, and \ref{h:hours} respectively.

We estimate \eqref{eq:logit} assuming that users visit \YA to check for new answers to their question every hour after the last answer arrived.
The maximum likelihood estimators are given in Table \ref{tab:logitBall}.  All parameter estimates are statistically significant at the 0.001 level.  We also used a generalized additive model~\cite{learning} to fit the data, which suggested that the assumed linearity in \eqref{eq:logit} is a good fit for the data.

It is worth noting that our results do not heavily depend on our assumption that users check for new answers every hour.  In particular, we get similar estimates for $\beta_1$, $\beta_2$ and $\beta_3$, if we assume that users check for answers every 2 hours, every 5 hours, or every 30 minutes.  For instance, if we assume that users check for new answers every 2 hours, we get $(\beta_1, \beta_2, \beta_3) = (0.026, 0.027, 0.022)$ instead of $(\beta_1, \beta_2, \beta_3) = (0.027, 0.028, 0.021)$.

\begin{table}
\begin{center}
\begin{tabular}{|c|c|}
  \hline
   & Estimate\\
  \hline
  $\alpha$  &-4.408*** (0.0603)\\
  $\beta_1$ & 0.027*** (0.0005)\\
  $\beta_2$ & 0.028*** (0.0002)\\
  $\beta_3$ & 0.021*** (0.0001)\\
  \hline
  Observations & 54,914\\
  \hline
\end{tabular}
\end{center}
\caption{The effect of the number of answers, the last inter-arrival time, and the time since the last answer on the probability of closing the question with $p$ as the dependent variable.  *, ** and *** denote significance at 1\%, 0.5\% and 0.1\% respectively.   Standard errors are given in parenthesis.}
\label{tab:logitBall}
\end{table}

We can now draw qualitative conclusions by considering the signs of the estimated coefficients.
We use the fact that the sign of a coefficient gives the sign of the corresponding marginal effect (since the logistic function is increasing).

First, the probability of closing the question is greater when more answers have arrived, which supports Hypothesis \ref{h:number}.  This implies that the marginal benefit of having one additional answer decreases as the number of answers increases, and is consistent with Table \ref{tab:cor} of Section \ref{sec:elapsed}.

Second, the probability of closing the question is greater when the last inter-arrival time is greater, which supports Hypothesis \ref{h:interarrival}.  The inverse of the inter-arrival time gives the rate at which answers arrive.  Thus, when the last inter-arrival time is large, i.e., there is a large time gap between the last answer and the answer before it, the user may expect that he will have to wait a long time until he receives the next answer.  This perceived high cost of waiting may encourage the user to close the question sooner when the last inter-arrival time is large.

Third, the probability of closing the question is greater when more time has elapsed since the last answer, which supports Hypothesis \ref{h:hours}.  As more time elapses since the last answer, the uncertainty increases, since the user does not know when the next answer will arrive.
The increased uncertainty may lead him to expect a greater waiting cost until the next answer arrives.  In turn, this encourages the user to close the question.

\subsection{Utility and Cost}
\label{sec:ut_est}

The following lemma establishes a specific quadratic form for the utility function.
\begin{lemma}
\label{l:utility}
If \eqref{eq:utility} holds, then
\begin{equation}
\label{eq:ut}
u(n) = \left(\alpha_u + \frac{\beta_1}{2}\right) n - \frac{\beta_1}{2} n^2 + u(0).
\end{equation}
\end{lemma}

\begin{proof}
If \eqref{eq:utility} holds, then
\begin{align*}
u(n) &=u(0)+\sum_{i=0}^{n-1} (\alpha_u - \beta_1 i) \\
     &= \left(\alpha_u + \frac{\beta_1}{2}\right) n - \frac{\beta_1}{2} n^2 + u(0)
\end{align*}
\end{proof}

We observe that the utility function given in \eqref{eq:ut} is concave on $[0, \infty)$ for any value of $\alpha_u$ as long as $\beta_1 > 0$, which is the case for our dataset since $\beta_1 = 0.027$ (see Table \ref{tab:logitBall}).
Moreover, for any fixed $\beta_1 > 0$ and $\alpha_u > 0$, the utility is unimodal: it is initially increasing (for $n < \lfloor\alpha_u/\beta_1 +0.5\rfloor$) and then decreasing (for $n > \lceil\alpha_u/\beta_1 +0.5\rceil$).  The latter may occur due to information overload; after a very large number of answers, the benefit of having one more answer may be so small that the cost of reading it exceeds the benefit, thus creating a disutility to the user.
Nevertheless, since questions at \YA rarely get a very large number of answers, the utility function given by \eqref{eq:ut} may be increasing throughout the domain of interest if $\alpha_u/\beta_1$ is sufficiently large.

We note that from \eqref{eq:logit} we estimate $\beta_1$ and $\alpha$, but cannot estimate $\alpha_u$.
For the sake of illustration, we plot the estimated utility $u(n)$ for various values of $\alpha_u$ in Figure \ref{fig:utility}, assuming that $u(0) = 0$. A reasonable domain to consider is $[0,50]$, since questions rarely get more than 50 answers.
We observe that when $\alpha_u$ is small (e.g., $\alpha_u = 1$), then the estimated $u(n)$ is decreasing for large values of $n$ within the $[0,50]$ region; suggesting an information overload effect.  On the other hand, for $\alpha_u \in \{2,3,4\}$, the estimated utility function is increasing throughout $[0,50]$.  Moreover, as $\alpha_u$ increases, the curvature of the estimated utility decreases, something we can also conclude from \eqref{eq:ut}.

\begin{figure}
\begin{center}
		\includegraphics[width=0.5\textwidth]{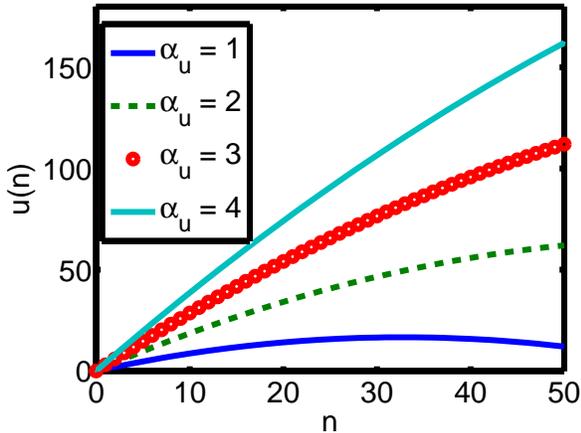}
\end{center}
\vspace{-1ex}
\caption
{Estimated $u(n)$ for $\alpha_u \in \{1,2,3,4\}$.}
\vspace{-2ex}
\label{fig:utility}
\end{figure}

We next consider the cost of waiting.  Equation \eqref{eq:cost} suggests that the expected cost increases linearly in both the last inter-arrival time $l$ and the time since the last answer arrival $w$.  However, it is not possible to get a specific form for the function $c(t)$ in the way we did for $u(n)$ in Lemma \ref{l:utility}, because we do not have any information on $T(l,w)$.  If we make assumptions on $T(l,w)$, we can draw conclusions about $c(t)$.  For instance, if we assume that the user is using a single estimate $\tau(l,w)$ on the time until the next answer arrives, then \eqref{eq:cost} implies that
\[c(\tau(l,w)) = \alpha_c + \beta_2 l + \beta_3 w.\]
Moreover, if the estimate $\tau(l,w)$ is linear in $l$ and $w$ we conclude that the cost of waiting is linear.  On the other hand, a concave cost would be consistent with a convex estimate, and a convex cost would be consistent with a concave estimate.

\section{Assessing Quality}
\label{sec:quality}

Our previous analysis considers how the decision problem of the user depends on the number of answers and time.  There is a third aspect that affects the user's decision to close his question: the quality of the answers that have arrived up to now.
In this section, we use an alternative model that incorporates quality, but does not incorporate time and the number of answers in the detail of Section \ref{sec:tradeoff}.  Our approach here is inspired by~\cite{surfing, surfing_science}.

Let $X_n$ be the value of the $n$-th answer.  This is in general subjective, and depends on the asker's interpretation and assessment. We assume that the value of an answer depends on both its quality and on the time that the user had to wait to get it.  For instance, $X_n$ may be negative if the waiting time was very large and the answer was not good (according to the user's judgement).
We model the values of the answers as a random walk, and assume that
\begin{equation}
\label{eq:RW}
X_{n+1} = X_n + Z_n,
\end{equation}
where the random variables $Z_n$ are independent and identically distributed.  For instance, if the user just got a high quality answer, he believes that the next answer will most likely also have high quality.  Similarly, if the user did not have to wait long for an answer, he expects that the next answer will probably arrive soon.  We note that \eqref{eq:RW} is consistent with the availability heuristic~\cite{availability}, which leads individuals to judge the frequency of an event by how easily they can bring an instance to mind.

Every time that a user sees an answer, he derives utility that is equal to the answer's value.  We assume that the user discounts the value he receives from future answers according to a discount factor $\delta$.
Let $V(x)$ be the maximum infinite 
horizon value when the value of the last answer is $x$.  Then,
\[V(x) = x  + \max\{0, \delta \cdot E(V(x+Z))\}.\]
In particular, the user decides to close the question if the value of closing exceeds the value of waiting for an additional answer.  If he closes the question, the user gets no future answers, and thus gets future value equal to 0.  On the other hand, if the user does not close the question, he gets value $E(V(x+Z))$ in the future, which he discounts by $\delta$.  Depending on which term is greater, the user decides whether to close the question or not.

We observe that $V(x)$ is increasing in $x$, which implies that $E(V(x+Z))$ is increasing in $x$.  We conclude that there exists a threshold $x^*$ such that it is optimal for the user to stop (i.e., close the question) when the value of the last answer is smaller than $x^*$ and to continue when the value of the last answer is greater than $x^*$.  The threshold $x^*$ satisfies $E(V(x^*+Z)) = 0$.

From an initial answer value, the user waits for additional answers, with values following a random walk as specified by \eqref{eq:RW}, until the value of an answer first hits the threshold value.  Thus, the number of answers until the user terminates the search is a random variable.  In the limit of true Brownian motion, the first passage times are distributed according to the inverse Gaussian distribution~\cite{surfing}. Then, the probability density of the number of answers to a question is given by
\begin{equation}
\label{eq:invG}
f(x) = \sqrt{\frac{\lambda}{2 \pi}} x^{-3/2} \exp\left(-\frac{\lambda}{2 \mu^2 x}(x-\mu)^2\right),
\end{equation}
where $\mu$ is the mean and $\lambda$ is a scale parameter.  We note that the variance is equal to $\mu^3/\lambda$.

We use the dataset to test the validity of \eqref{eq:invG}.  We find that the maximum likelihood inverse Gaussian has $\mu = 6.1$ and $\lambda = 5.8$.  Figure \ref{fig:cdf} shows the empirical and fitted cumulative distribution functions.  We observe that the inverse Gaussian distribution is a very good fit for the data.

An important property of the inverse Gaussian distribution is that for large variance, the probability density is well-approximated by a straight line with slope -3/2 for larger values of $x$ on a log-log plot; thus generating a Zipf-like distribution.  This can be easily seen by taking logarithms on both sides of \eqref{eq:invG}.  In Figure \ref{fig:loglog} we plot the frequency distribution of the number of answers on log-log scales.  We observe that the slope at the tail is approximately -3/2.

\begin{figure}
\begin{center}
		\includegraphics[width=0.5\textwidth]{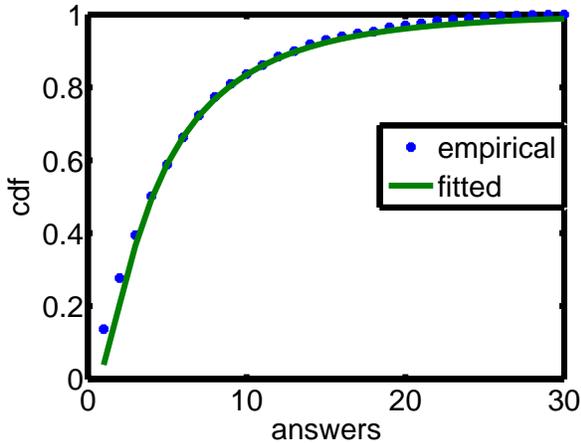}
\end{center}
\vspace{-1ex}
\caption
{Empirical and inverse Gaussian fitted cumulative distributions.
The points are the empirical cumulative distribution function of the number of answers.  The curve is the cumulative distribution function of the maximum likelihood inverse Gaussian.}
\vspace{-2ex}
\label{fig:cdf}
\end{figure}

\begin{figure}
\begin{center}
        \includegraphics[width=0.5\textwidth]{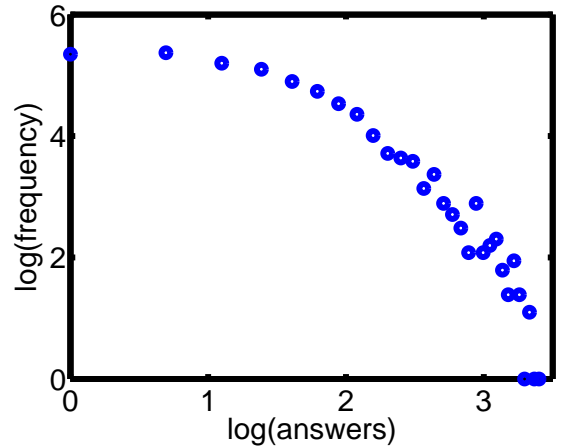}
\end{center}
\vspace{-1ex}
\caption
{The frequency distribution of the number of answers on log-log scales.}
\vspace{-2ex}
\label{fig:loglog}
\end{figure}

\section{Conclusion}
\label{sec:conclusion}

This paper empirically studies how people behave when they face speed-accuracy tradeoffs.
We have taken two complementary approaches.

Our first approach is to study the speed-accuracy tradeoff by using the number of answers as a proxy for accuracy.  In particular, we assume that the user approximates the accuracy of his aggregate answer by the number of answers that his question gets.  Thus, the user faces the following tradeoff: he prefers more to less answers, but does not want to wait too long.  We analyze \YA data to identify and quantify this tradeoff.  We find that users are willing to wait longer to obtain one additional answer when they have only received a small number of answers; this implies  decreasing marginal returns in the number of answers, or equivalently, a concave utility function.  We then estimate the utility function from the data.

Our second approach focuses on how users assess the qualities of the individual answers without explicitly considering the cost of waiting.  We assume that users make a sequence of decisions to wait for another answer, deciding to wait as long as the current answer exceeds some threshold in value.  Under this model, the probability distribution for the number of answers that a question gets is an inverse Gaussian, which is a Zipf-like distribution.  We use the data to validate this conclusion.

It remains an open question how to combine these two approaches in order to study the speed-accuracy tradeoff by jointly considering the number of answers, their qualities, and their arrival times.

We conclude by noting that our results could be used by \YA or other question-answering sites to prioritize the way questions are shown to potential answerers in order to maximize social surplus.
The key observation is that a question receives answers at a higher rate when it is shown on the first page at \YAc.  On the other hand, the rate at which answers are received also depends on the quality of the question.
Using appropriate information about these rates as well as the utility function estimated in this paper, the site can position open questions with the objective of maximizing the sum of users' utilities.

\section{Acknowledgements}

We gratefully acknowledge Eugene Agichtein for providing the dataset~\cite{datasetB}, as well as detailed information on how the data was collected.

\bibliographystyle{abbrv}
\bibliography{ref}

\end{document}